\documentclass[a4paper,onecolumn,11pt,unpublished]{quantumarticle}
\pdfoutput=1
\usepackage[utf8]{inputenc}
\usepackage[english]{babel}
\usepackage[T1]{fontenc}
\usepackage{amsmath}
\usepackage{hyperref}
\usepackage{amssymb}
\usepackage{mathtools}
\usepackage{braket}
\usepackage{thmtools}
\usepackage{booktabs}
\usepackage{multirow}

\usepackage{tikz}
\usetikzlibrary{arrows.meta}
\usepackage{lipsum}

\usepackage{amsthm}
\newtheorem{theorem}{Theorem}
\newtheorem{lemma}[theorem]{Lemma}
\newtheorem{corollary}[theorem]{Corollary}
\newtheorem{definition}[theorem]{Definition}

\newcommand{\ent}{\mathcal{H}}

\newcommand{\maintheoremtext}{
    Let $H$ be an $n$-qubit $k$-local Hamiltonian and $E_d$ be defined as in Equation~(\ref{eq:ed_def}). For any $\varepsilon =\Omega(1/\mathrm{poly}(n))$, there exists a quantum algorithm that, with high probability, prepares an $n$-qubit quantum state $\rho$ satisfying
    $$ \operatorname{Tr}[H\rho] \le E_d+\varepsilon M $$
    and outputs a corresponding energy estimate $\widehat{E}$ lying within the window
    $$ \widehat{E} \in [\lambda_0,\ E_d+\varepsilon M]. $$
    The running time of the algorithm is
    $$O^*\left( 2^{\frac{n}{2}\left(1 - \frac{1}{2}\ent\left(\frac{\varepsilon }{2^{d+2}k}\right)\right)} \right), $$
    where $\ent(\cdot)$ denotes the binary entropy function.
}
\newtheorem*{reptheorem}{Theorem \ref{th:main_intro} (repeated)}

\begin{document}
	\author{Ranitha Mataraarachchi}
	\affiliation{Graduate School of Mathematics, Nagoya University, Japan}
	\orcid{0009-0003-6053-6667}
    
	\title{An Entropy-Governed Speedup for Quantum Algorithms on Local Hamiltonians}
	\author{François Le Gall}
	\affiliation{Graduate School of Mathematics, Nagoya University, Japan}

	\author{Suguru Tamaki}
	\affiliation{Graduate School of Information Science, University of Hyogo, Japan}
	\maketitle
	
	\begin{abstract}       
        Low-energy estimation and state preparation for general $k$-local Hamiltonians are fundamental challenges in quantum complexity theory. For constant relative accuracy, Buhrman et al.~(PRL 2025) recently broke the natural Grover bound $O(2^{n/2})$, where $n$ denotes the number of qubits, for both problems. In this paper, for any sufficiently small parameter $d\ge 0$, we present an even faster quantum algorithm that outputs a quantum state with energy bounded by the minimum energy over all depth-$d$ states (i.e., states obtained by applying a depth-$d$ circuit to the all-zero state), together with an estimate of this energy. For the class of Hamiltonians with depth-$d$ ground states, our algorithm furthermore achieves exactly the same energy guarantees as Buhrman et al. Our results also provide insight into the distinction between strongly entangled states and those admitting efficient classical descriptions.
	\end{abstract}
	
	\section{Introduction}
	
		The study of quantum many-body systems is fundamentally related to the properties of local Hamiltonians. Estimating the ground-state energy and preparing low-energy states of such systems are central tasks in condensed matter physics, quantum chemistry, and quantum information. Consider an $n$-qubit $k$-local Hamiltonian 
	$$H = \sum_{\alpha=1}^m h_\alpha,$$ 
	where each interaction term $h_\alpha$ is a Hermitian operator acting on at most $k$ qubits, for some constant $k$. Let $\{\lambda_i\}_{i=0}^{2^{n}-1}$ denote the eigenvalues of $H$ arranged in non-decreasing order, with $\lambda_0$ being the ground-state energy. For systems with strict geometric locality, such as spins on a low-dimensional lattice, constant-factor approximations to the ground-state energy can often be found efficiently by severing the system into small patches \cite{bravyi}. However, these techniques fail for systems with long-range couplings. In this work, we consider general all-to-all interacting Hamiltonians and do not rely on any assumptions of geometric locality or underlying spatial structure.
	
	Estimating the ground state energy and, more generally, preparing the ground state, of such general $k$-local Hamiltonian is a central problem in quantum complexity theory. It is well-established that determining $\lambda_0$ up to inverse-polynomial precision is QMA-hard~\cite{qma1,qma2,qma5,qma3,qma4}, meaning that even a quantum computer is unlikely to solve it efficiently in worst-case scenarios. While heuristic variational algorithms, such as the Variational Quantum Eigensolver (VQE)~\cite{vqe}, can be executed in polynomial time, they are fundamentally local optimization methods. As such, their optimization landscapes are highly susceptible to barren plateaus and getting trapped in local minima, offering no rigorous convergence guarantees for general many-body systems~\cite{barren,vqeNP}.
	
	This highlights the need for quantum algorithms with provable asymptotic speedups for low-energy estimation and state preparation. To formalize the notion of estimation, define the total interaction strength of the Hamiltonian as 
    $$M \coloneqq \sum_{\alpha=1}^m \Vert h_\alpha \Vert,$$ 	
    where $\Vert\cdot\Vert$ denotes the spectral norm. In this work, we focus on relative accuracy (also called extensive accuracy): the goal is to prepare a state and estimate its energy, targeting a value bounded by a threshold plus an error $\varepsilon M$. Approximating ground states up to a relative error represents a physically relevant regime for macroscopic systems. This regime is also closely connected to the approximation scale considered in the Quantum PCP conjecture \cite{pcp, bh}.
	
	In the absence of a specific Hamiltonian structure, the most direct approach to preparing a low-energy state relies on unstructured quantum search (Grover's algorithm), which operates in time $O^\ast(2^{n/2})$ \cite{grover, brassard, grov1, grov2, grov3}.\footnote{In this paper, the notation $O^\ast(\cdot)$ removes the polynomial factors in $n$.} Breaking this natural Grover speedup barrier for general local Hamiltonians was recently achieved by the work of Buhrman et al.~\cite{buhr1}. This work established the first quantum algorithms to estimate the ground energy and prepare a state near that energy, up to a relative error, strictly faster than $O^\ast(2^{n/2})$:
	
	\begin{theorem}[Buhrman et al.~\cite{buhr1}]\label{th:buhrman}
		Let $H$ be an $n$-qubit $k$-local Hamiltonian. For any $\varepsilon =\Omega(1/\mathrm{poly}(n))$, there exists a quantum algorithm that, with high probability, prepares an $n$-qubit quantum state $\rho$ satisfying
        $$ \operatorname{Tr}[H\rho] \le \lambda_0+\varepsilon M $$
        and outputs a corresponding energy estimate $\widehat{E}$ lying within the window
        $$ \widehat{E} \in [\lambda_0,\ \lambda_0+\varepsilon M]. $$
        The running time of the algorithm is
        $$O^*\left( 2^{\frac{n}{2}\left(1 - \frac{\varepsilon}{2k+\varepsilon}\right)}\right).\footnote{In the case of low-energy estimation, \cite{buhr1} also shows how to improve the term $\frac{\varepsilon}{2k+\varepsilon}$ to $\frac{\varepsilon}{k+\varepsilon}$, at the price of slightly increasing the window.}$$
	\end{theorem}
	
	\subsection{Our Contributions}
	
	In this work, we demonstrate that by targeting a different energy threshold, we can achieve an exponential runtime scaling that is strictly faster than the bounds established by Buhrman et al. In addition, for Hamiltonians whose ground state is promised to lie within a restricted class of states, our algorithms achieve improved runtimes while matching the energy guarantees of Buhrman et al.
	
	Rather than targeting the absolute ground energy $\lambda_0$, we focus on the energy attainable by a restricted class of states that can be created by low-depth quantum circuits. Such low-depth states correspond to the ``trivial states'' considered, for instance, in the NLTS literature \cite{nlts,eldar,trivial}.
    
    A depth-$d$ state $\ket{\phi_d}$ is defined as a state obtained by applying a quantum circuit $U_d$ of depth $d$, composed of two-qubit gates, to the all-zero state $\ket{0}^{\otimes n}$.
	We define $E_d$ as the minimum energy over all depth-$d$ states:
	\begin{align*}
		E_d \coloneqq \min_{\ket{\phi_d} \in \mathcal{T}_d} \bra{\phi_d} H \ket{\phi_d},
	\end{align*}
	where $\mathcal{T}_d$ denotes the set of all depth-$d$ states. 
    This is equivalently expressed as
	\begin{align}\label{eq:ed_def}
		E_d = \min_{U_d \in \mathcal{C}_d} \bra{0}^{\otimes n} U_d^\dagger H U_d \ket{0}^{\otimes n},
	\end{align}
	where $\mathcal{C}_d$ is the set of all depth-$d$ quantum circuits composed of two-qubit gates. This threshold captures the limitations of the energy achievable by low-depth states. 
	
	Our main technical contributions are faster algorithms for (i) preparing a low-energy state relative to $E_d$ and (ii) estimating its energy:
    \begin{theorem}[Faster Low-Energy State Preparation and Estimation]\label{th:main_intro}
        \maintheoremtext
    \end{theorem}
	
	Our algorithms and the work of Buhrman et al.~present a fundamental trade-off between the depth of the energy approximation and the speed of the algorithm. The algorithms of Buhrman et al.~target a neighborhood bounding the absolute ground energy $\lambda_0$. Their execution time achieves a super-quadratic speedup characterized by a rational function $\left(\frac{\varepsilon}{2k+\varepsilon}\right)$ in the exponent. In contrast, our algorithms target the higher energy threshold $E_d$. Because we exploit the density of the states associated with this threshold, our asymptotic speedup is governed by the binary entropy function $\ent(\cdot)$.

    For small values of the ratio $\varepsilon/k$, the rational function $\left(\frac{\varepsilon}{2k+\varepsilon}\right)$ scales linearly. However, the binary entropy function introduces a logarithmic factor, scaling as $\Theta(\frac{\varepsilon}{k} \log \frac{k}{\varepsilon})$. Because $\ent(\cdot)$ grows strictly faster than the rational function due to this logarithmic advantage, our algorithms execute strictly faster in the regime where
    $$\frac{1}{2}\ent\left(\frac{\varepsilon}{2^{d+2}k}\right) \geq \frac{\varepsilon}{2k+\varepsilon},
    $$
    which implies that
    $$d = O\left(\log\log\left(\frac{k}{\varepsilon}\right)\right).$$
    See Table~\ref{tab:comparison} for a numerical comparison between the exponents.

    Importantly, if the ground state of $H$ is itself a low-depth state, then $E_d = \lambda_0$. For such Hamiltonians, our algorithms achieve the same energy guarantee as Buhrman et al., while attaining an improved runtime. 

    \paragraph*{Related work.}
    Recently, Buhrman et al.~\cite{buhr2} showed how to apply the ideas from \cite{buhr1} to derive new \emph{classical} algorithms for classical constraint satisfaction problems (e.g, MAX-$k$-SAT). For such classical optimization problems, the ground state is classical.
    Our approach can be seen as a generalization of this approach to quantum Hamiltonians. Note that, contrary to \cite{buhr2}, we consider \emph{quantum} algorithms and also focus on the preparation of low-energy states.
	
	\begin{table}[h]
		\centering
		\renewcommand{\arraystretch}{1.3}
		\begin{tabular}{cc ccc}
			\toprule
			& & \multicolumn{3}{c}{\textbf{Time Complexity Exponent Coefficient ($c$) for Est./ Prep.}} \\
			\cmidrule(lr){3-5} 
			$\boldsymbol{k}$ & $\boldsymbol{\varepsilon}$ & \textbf{Buhrman et al.~\cite{buhr1}} & \textbf{Our Work ($d=0$)} & \textbf{Our Work ($d=1$)} \\
			\midrule
			\multirow{4}{*}{$3$} 
			& $1/8  $ & $0.4897959$ & $0.4791143$ & $0.4882501$ \\
			& $0.05 $ & $0.4958678$ & $0.4902640$ & $0.4946104$ \\
			& $0.01 $ & $0.4991681$ & $0.4975686$ & $0.4986801$ \\
			& $0.001$ & $0.4999167$ & $0.4996876$ & $0.4998334$ \\
			\midrule
			\multirow{4}{*}{$4$} 
			& $1/8  $ & $0.4923077$ & $0.4835214$ & $0.4907814$ \\
			& $0.05 $ & $0.4968944$ & $0.4923732$ & $0.4957955$ \\
			& $0.01 $ & $0.4993758$ & $0.4981115$ & $0.4989776$ \\
			& $0.001$ & $0.4999375$ & $0.4997592$ & $0.4998718$ \\
			\midrule
			\multirow{4}{*}{$10$} 
			& $1/8  $ & $0.4968944$ & $0.4923732$ & $0.4957955$ \\
			& $0.05 $ & $0.4987531$ & $0.4965357$ & $0.4981115$ \\
			& $0.01 $ & $0.4997501$ & $0.4991620$ & $0.4995497$ \\
			& $0.001$ & $0.4999750$ & $0.4998954$ & $0.4999446$ \\
			\bottomrule
		\end{tabular}
		\vspace{0.2cm}
		\caption{Comparison of the time complexity exponent $c$, where algorithms run in time $O^*(2^{cn})$. A smaller $c$ indicates a faster execution. Buhrman et al.~\cite{buhr1} target a relative error above the absolute ground energy $\lambda_0$, whereas our work executes faster by targeting $E_d$.}
		\label{tab:comparison}
	\end{table}

	\subsection{Preliminaries and Notation}\label{sec: prelim}
	
	As already mentioned, we consider an $n$-qubit $k$-local Hamiltonian $H=\sum_{\alpha=1}^{m}h_\alpha$ with eigenvalues $\lambda_0 \le \lambda_1 \le \dots \le \lambda_{2^n-1}$ and total interaction strength $M = \sum_{\alpha=1}^m \Vert h_\alpha \Vert.$
	
	To analyze the density of low-energy states, we define the \emph{cumulative spectral count}, denoted $N(E)$, as the number of eigenstates with energy at most $E$:
	\begin{align*}
		N(E) \coloneqq \big|\{ i \in \{0, \dots, 2^n - 1\} \mid \lambda_i \le E \}\big|.
	\end{align*}
	
	Recall our target energy threshold $E_d$, defined in Equation~(\ref{eq:ed_def}) as the minimum energy achievable by applying any depth-$d$ circuit to the all-zero state $\ket{0}^{\otimes n}$. Let $U_d^* \in \mathcal{C}_d$ be an optimal depth-$d$ circuit that achieves this minimum. Therefore, $E_d$ can now be expressed as
	\begin{align*}
		E_d 
		=
		\bra{0}^{\otimes n}(U_d^*)^\dagger HU_d^* \ket{0}^{\otimes n}. 
	\end{align*}
	
	We can absorb the optimal circuit $U_d^*$ into the Hamiltonian. We define the \emph{effective Hamiltonian} $H_d$ as:
	\begin{align*}
		H_d \coloneqq (U_d^*)^{\dagger} H U_d^*.
	\end{align*}
	Consequently, we can express $E_d$ as an expectation over this effective Hamiltonian:
	\begin{align*}\label{eq: Phi}
		E_d = \bra{0}^{\otimes n}H_d \ket{0}^{\otimes n}. 
	\end{align*}
	
	Without loss of generality (see~\cite[Section 2.3]{anshu} for a justification based on the strict light-cone argument), the conjugation of a local term by a low-depth circuit strictly bounds the spread of correlations. Specifically, any Hamiltonian term $(U_d^*)^{\dagger} h_\alpha U_d^*$ expands its support by at most a factor of $2^d$.\footnote{The support $k$ of a local interaction term at most doubles at each layer of the two-qubit gate circuit.} Therefore, the effective Hamiltonian $H_d$ is a $2^d k$-local Hamiltonian, for some value $2^d k$ which we will define as $K$.
	
	Because $H_d$ and $H$ are unitarily equivalent by definition, both operators possess identical eigenvalue spectra (they are isospectral). Therefore, establishing a lower bound on the cumulative spectral count $N(E)$ for the effective Hamiltonian $H_d$ directly and equivalently bounds the cumulative spectral count for the original Hamiltonian $H$.

    Since $H$ is Hermitian, and therefore \emph{normal}, it admits a spectral decomposition $H=\sum_i \lambda_i \ket{\phi_i}\bra{\phi_i}$. Let $E \in \mathbb{R}$ be an energy threshold. We define the low-energy projector as $\Pi_{\le E} = \sum_{i : \lambda_i \le E} \ket{\phi_i}\bra{\phi_i}$. The overlap $\gamma$ of a quantum state $\ket{\psi}$ with the subspace of energies below $E$ is given by the expectation value of this projector:
    $$ \gamma 
    = 
    \bra{\psi} \Pi_{\le E} \ket{\psi} 
    = 
    \sum_{i : \lambda_i \le E} \braket{\psi\vert\phi_i}\braket{\phi_i\vert\psi}
    =
    \sum_{i : \lambda_i \le E} \vert\braket{\psi\vert\phi_i}\vert^2.
    $$
	
	As already mentioned, for asymptotic scaling, we use the notation $f(n) = O^\ast(g(n))$ to indicate that $f(n) \le g(n) \cdot \mathrm{poly}(n)$, and $f(n) = \Omega^\ast(g(n))$ to indicate that $f(n) \ge g(n)\cdot\mathrm{poly}^{-1}(n)$, thereby suppressing polynomial factors in $n$ for both the lower and upper bounds, respectively.

	\subsection{Organization of the Paper}
	The remainder of this paper is organized as follows. In Section~\ref{bigsec2}, we establish the theoretical foundation of our work by analyzing the density of low-energy states. We first prove that the effective $K$-local Hamiltonian $H_d$ possesses an exponentially dense spectrum below the energy threshold $E_d + \mu M$, for some $\mu>0$, by showing the existence of a large family of orthonormal, low-energy product states. Because $H_d$ and $H$ are isospectral, we then extend this result via a corollary to prove that the original Hamiltonian $H$ shares this identical dense low-energy spectrum, which is the main result in Section~\ref{bigsec2}. 
	
	Building upon this spectral property, Section~\ref{bigsec3} presents our quantum algorithms for faster low-energy estimation and state preparation. By utilizing the maximally entangled state to guarantee an exponential overlap with the dense low-energy subspace of $H$, we derive the accelerated time complexity bounds. 
    
	Finally, Section~\ref{bigsec4} provides a brief conclusion and summarizes the broader implications of this work.
	
	\section{The Dense Low-Energy Spectrum}\label{bigsec2}
	In this section, we establish the main theoretical foundation of our work, which is the exponential lower bound on the density of low-energy eigenstates for general $k$-local Hamiltonians. Specifically, we prove that the cumulative spectral count $N(E_d + \mu M)$ is exponentially large. This result serves as the theoretical engine that will allow our state preparation algorithm (presented in Section~\ref{bigsec3}) to bypass the standard Grover search barrier. To establish this bound, we exploit the geometric structure of the product-state manifold.

	\subsection{Proof Strategy}
	To derive our lower bound on $N(E_d + \mu M)$, we employ the following strategy:

	\begin{enumerate}
		\item \textbf{Existence of Many Orthonormal States (Section~\ref{sec1}):} We identify a set of qubit sites with low energy contribution and show that perturbing these sites in $\ket{0}^{\otimes n}$ yields a large family of orthonormal product states.
		\item \textbf{Energy Bounding (Section~\ref{sec2}):} We determine the maximum number $r$ of qubit sites that can be perturbed while increasing the energy relative to $E_d$ by at most $\mu \eta M$, for some scaling factor $\eta \in [0,1]$.
		\item \textbf{Bounding the Low-Energy Spectrum Count (Section~\ref{sec3}):} Finally, having established the existence of a large set of orthonormal states with bounded energy, we invoke the Cauchy Interlacing Theorem to relate these states to the eigenstates of $H_d$, thereby establishing a lower bound on $N(E_d+\mu M).$
	\end{enumerate}
	
	\subsection{Existence of Many Mutually Orthogonal States}\label{sec1}
	First, we identify a set of qubit sites with low energy contribution, which we call the Quiet Set $Q_\delta$. The Hamiltonian can be expanded as $H_d=\sum_{\alpha=1}^{m}h'_\alpha$, where the local interaction terms $h'_\alpha$ each act on at most $K$ qubits. By the definition of the total interaction strength,
	\begin{align*}
		M = \sum_{\alpha=1}^m\Vert h_\alpha\Vert = \sum_{\alpha=1}^m\Vert (U_d^*)^{\dagger} h_\alpha U_d^*\Vert=\sum_{\alpha=1}^m\Vert h'_\alpha\Vert.
	\end{align*}
	Therefore, both $H$ and $H_d$ have the same total interaction strength. To quantify the effective Hamiltonian's interaction strength on each qubit site $s \in \{1,\dots,n\}$, we define the local energy contribution $e(s)$ as
	\begin{align*}
		e(s) \coloneqq \sum_{\alpha\in\mathcal{I}(s)}\Vert h'_\alpha\Vert,
	\end{align*}
	where $\mathcal{I}(s)$ is the set of interactions acting on $s$. 
	We define the sum of local energy contributions as
	\begin{equation}\label{eq:defL}
		L \coloneqq \sum_{s=1}^{n}e(s).
	\end{equation}
	
	\begin{definition}[The Quiet Set $Q_\delta$]
		For any constant $\delta \ge 1 + \frac{\mu M} {L} $, we define the quiet set as the subset of qubit sites that have a bounded local energy contribution relative to the average:
		$$Q_\delta=\left\{s\in\{1,\dots,n\} \;\middle|\; e(s)\leq \frac{\delta L}{n}\right\}.$$
	\end{definition}
	
	\begin{lemma}\label{lem: Lower}
		For any $\delta \ge 1 + \frac{\mu M}{L}$, the size of the quiet set satisfies $$\vert Q_\delta\vert \geq \frac{(\delta-1)n}{\delta}.$$
	\end{lemma}
	\begin{proof}
		The sum of the local energies for qubits strictly outside the quiet set must satisfy
		$$(n - |Q_\delta|) \frac{\delta L}{n} < \sum_{s \notin Q_\delta} e(s) \le L.$$
		Solving this inequality for $|Q_\delta|$ directly yields the result. $\qedhere$
	\end{proof}
	
	Next, we prove the existence of a family of orthonormal product states derived from $\ket{0}^{\otimes n}$ by perturbing the quiet set $Q_\delta$. See Figure~\ref{fig:rotations} for an illustration. We also provide a lower bound on the cardinality of this family. The states are defined by modifying subsets of qubits within $Q_\delta$, where the size of these subsets is bounded by $r$ (to be determined in Section~\ref{sec2}).

	\begin{center}
		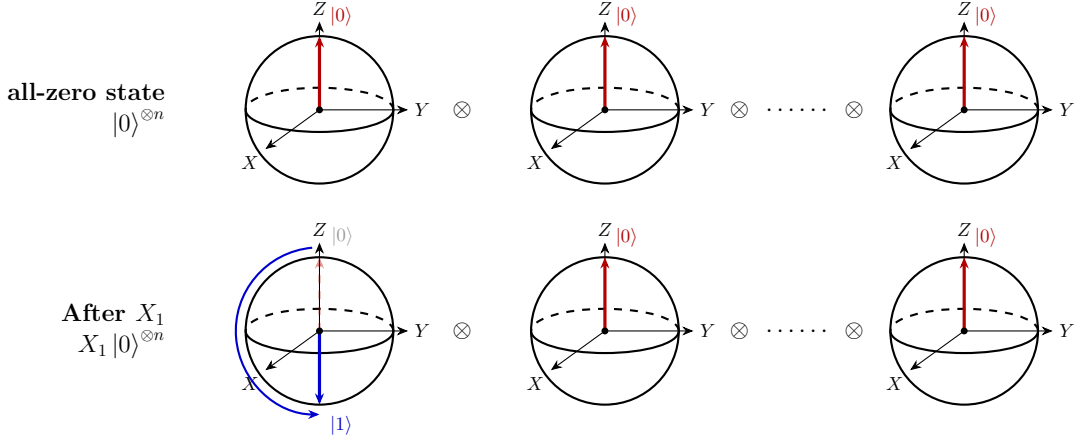
\begin{figure}[h]
			\centering
				\begin{tikzpicture}[
			scale=0.65, 
			transform shape,
			>={Stealth[length=1.5mm, width=1.2mm]}
			]
			
			\newcommand{\blochsphere}{
				\draw[thick] (0,0) circle (1.5);
				\draw[thick] (1.5,0) arc (0:-180:1.5 and 0.45);
				\draw[thick, dashed] (1.5,0) arc (0:180:1.5 and 0.45);
				
				\draw[->] (0,0) -- (0,1.8) node[above] {$Z$};
				\draw[->] (0,0) -- (-1.1,-0.8) node[below left] {$X$};
				\draw[->] (0,0) -- (1.8,0) node[right] {$Y$};
			}
			
			
			\node[left, align=right, font=\bfseries\Large] at (-3, 0) 
			{ all-zero state \\ $\ket{0}^{\otimes n}$};
			
			\begin{scope}[shift={(0,0)}]
				\blochsphere
				\draw[->, very thick, red!70!black] 
				(0,0) -- (0,1.5) node[above right=2pt] {$\ket{0}$};
				\fill (0,0) circle (2pt);
			\end{scope}
			
			\node at (2.9,0) {\Large $\otimes$};
			
			\begin{scope}[shift={(5.8,0)}]
				\blochsphere
				\draw[->, very thick, red!70!black] 
				(0,0) -- (0,1.5) node[above right=2pt] {$\ket{0}$};
				\fill (0,0) circle (2pt);
			\end{scope}
			
			\node at (9.7,0) {\Large $\otimes \,\,\,\cdots\cdots \,\,\,\otimes$};
			
			\begin{scope}[shift={(13.1,0)}]
				\blochsphere
				\draw[->, very thick, red!70!black] 
				(0,0) -- (0,1.5) node[above right=2pt] {$\ket{0}$};
				\fill (0,0) circle (2pt);
			\end{scope}

			
			\node[left, align=right, font=\bfseries\Large] at (-3.0, -4.5) 
			{After $X_1$ \\ $X_1 \ket{0}^{\otimes n}$};
			
			\begin{scope}[shift={(0,-4.5)}]
				\blochsphere
				
				\draw[->, thick, red!70!black, dashed, opacity=0.4] 
				(0,0) -- (0,1.5)
				node[above right=1mm, text=black, opacity=0.4] {$\ket{0}$};
				
				\draw[->, very thick, blue!80!black] 
				(0,0) -- (0,-1.5) node[below right=2pt] {$\ket{1}$};

				\fill (0,0) circle (2pt);
				
				\draw[->, thick, blue!80!black] (95:1.7) arc (95:270:1.7) node[midway, above left=1mm] {};
			\end{scope}
			
			\node at (2.9,-4.5) {\Large $\otimes$};
			
			\begin{scope}[shift={(5.8,-4.5)}]
				\blochsphere
				\draw[->, very thick, red!70!black] 
				(0,0) -- (0,1.5) node[above right=2pt] {$\ket{0}$};
				\fill (0,0) circle (2pt);
			\end{scope}

            	\node at (9.7,-4.5) {\Large $\otimes \,\,\,\cdots\cdots \,\,\,\otimes$};
			
			\begin{scope}[shift={(13.1,-4.5)}]
				\blochsphere
				\draw[->, very thick, red!70!black] 
				(0,0) -- (0,1.5) node[above right=2pt] {$\ket{0}$};
				\fill (0,0) circle (2pt);
			\end{scope}
			
		\end{tikzpicture}
		\caption{Visual representation of the local unitary perturbation applied to the first qubit, assuming it is in the quiet set.}
        \label{fig:rotations}
		\end{figure}
	\end{center}
	
	\begin{lemma}\label{lem: Orthgonality}
		For any subset $R \subseteq Q_\delta$, define the unitary $X_R \coloneqq \bigotimes_{s \in R} X_s,$ acting trivially on qubits outside $R$, where $X_s$ is the Pauli $X$ operator acting on the qubit site $s$. We define the perturbed state corresponding to the set $R$ as
		\begin{align*}
			\ket{\Phi_R} \coloneqq X_R \ket{0}^{\otimes n}.
		\end{align*}
		Then, for any two distinct subsets $R, R' \subseteq Q_\delta$, the corresponding states are orthogonal:
		$$\braket{\Phi_R | \Phi_{R'}} = 0.$$
	\end{lemma}
	\begin{proof}
		Consider two distinct subsets $R \neq R'$. There exists at least one qubit index $s$ in the symmetric difference $R \triangle R'$. Without loss of generality, assume $s \in R \setminus R'$.
		
		At site $s$, the state $\ket{\Phi_R}$ contains the qubit $X_s \ket{0} = \ket{1}$, whereas the state $\ket{\Phi_{R'}}$ contains the unperturbed qubit $\ket{0}.$
		
		Because the global inner product of product states factorizes into the product of local inner products, we have
		$$\braket{\Phi_R | \Phi_{R'}} = \braket{1 | 0} \cdot \prod_{t \neq s} \braket{\phi_t^{(R)} | \phi_t^{(R')}} = 0 \cdot \prod_{t \neq s} \braket{\phi_t^{(R)} | \phi_t^{(R')}} = 0,$$
		where $\ket{\phi_t^{(R)}}$ denotes the single-qubit state at site $t$ in $\ket{\Phi_R}.$ Thus, the states are orthogonal. \(\qedhere\)
	\end{proof}
	
	Using the above lemma, we can now lower-bound the number of such orthogonal states available to us.
	\begin{corollary}\label{cor: Sr Lower}
		Let $\mathcal{S}_r \coloneqq \{\ket{\Phi_R} : R \subseteq Q_\delta, |R|\le r\}$ be the family of orthonormal states defined in Lemma~\ref{lem: Orthgonality}. Then, for any $r\leq \vert Q_\delta\vert/2$, the size of this family satisfies
		
		$$\vert\mathcal{S}_r\vert = \Omega^*\left(\max_{\delta \ge 1 + \frac{\mu M}{L}}\left\{2^{\ent\left(\frac{r\delta}{(\delta-1)n}\right)\frac{(\delta-1)n}{\delta}}\right\} \right).$$
	\end{corollary}
	
	\begin{proof}By Lemma~\ref{lem: Orthgonality}, any two distinct states $\ket{\Phi_R}$ and $\ket{\Phi_{R'}}$ are orthogonal whenever $R \neq R'$. Therefore, the size of the set $\mathcal{S}_r$ is exactly equal to the number of distinct subsets $R \subseteq Q_\delta$ with cardinality $\vert R\vert \le r.$ The number of such subsets is given by the sum of binomial coefficients:
		$$\vert\mathcal{S}_r\vert = \sum_{i=0}^{r} \binom{\vert Q_\delta\vert}{i}.$$
		
		Using the standard lower bound for the sum of binomial coefficients up to $r \le n/2$, specifically $\sum_{i=0}^r \binom{n}{i} \ge 2^{n \ent(r/n)} / (n+1)$, we have
		\begin{align*}
			\vert\mathcal{S}_r\vert \geq \binom{\vert Q_\delta\vert}{r} \geq \frac{2^{\ent\left(\frac{r}{\vert Q_\delta\vert}\right)\vert Q_\delta\vert}}{\vert Q_\delta\vert+1} = \Omega^*\left(2^{\ent\left(\frac{r}{\vert Q_\delta\vert}\right)\vert Q_\delta\vert}\right).
		\end{align*}
		
		We now utilize the monotonicity of the entropy bound. The function $f(\alpha) = \alpha \ent(r/\alpha)$ is monotonically increasing with respect to $\alpha$ for $r \le \alpha/2$. Recall from Lemma~\ref{lem: Lower} that the size of the quiet set is lower-bounded by $\vert Q_\delta\vert \geq \frac{(\delta-1)n}{\delta}$.
		
		Substituting this lower bound for $\vert Q_\delta\vert$ into the exponent implies
		$$\vert\mathcal{S}_r\vert = \Omega^*\left(2^{\ent\left(\frac{r\delta}{(\delta-1)n}\right)\frac{(\delta-1)n}{\delta}}\right).$$
		Optimizing over all valid $\delta \ge 1 + \frac{\mu M}{L}$ yields the statement of the Corollary. \(\qedhere\)
	\end{proof}
	
	\subsection{Energy Analysis of Perturbed States}\label{sec2}
	Recall from Section~\ref{sec1} that the value of $r$ was left to be determined. In this subsection, we establish the maximum size $r$ for the subsets $R$ such that the resulting orthonormal product states remain within a bounded energy window near $E_d$. Specifically, we determine the value of $r$ required to ensure that any state $\ket{\Phi_R}$ in the family $\mathcal{S}_r$ satisfies:
	\begin{align*}
		\vert\bra{\Phi_R} H_d \ket{\Phi_R} - E_d\vert \le \mu\eta M.
	\end{align*}
	where $\mu>0$ and $\eta\in[0,1]$.
	\begin{lemma}\label{lem: Energy}
		Let $\mathcal{S}_r=\{\ket{\Phi_R} : |R|\le r\}$ be the family of orthonormal states identified in Lemma~\ref{lem: Orthgonality}.
		
		If the size of the subset is bounded by $|R|\le r$, where
		$$r=\left\lfloor \frac{\mu \eta M n}{2\delta L} \right\rfloor,$$
		then the energy of the perturbed state satisfies
		$$\vert\bra{\Phi_R} H_d \ket{\Phi_R} - E_d\vert \le \mu\eta M.$$
	\end{lemma}
	
	\begin{proof}
		We evaluate the energy difference $\Delta \coloneqq \vert\bra{\Phi_R}H_d\ket{\Phi_R} - E_d\vert$. 
        
        Recall that $\ket{\Phi_R} = X_R \ket{0}^{\otimes n}$. Substituting this into the energy expression:
		\begin{align*}
			\Delta &= \vert\bra{0}^{\otimes n} X_R^\dagger H_d X_R \ket{0}^{\otimes n} - \bra{0}^{\otimes n} H_d \ket{0}^{\otimes n}\vert.
		\end{align*}
		Let $T$ denote the set of indices for interaction terms $h'_\alpha$ that act on at least one qubit in $R$. We decompose the Hamiltonian into terms involved with $R$ and those that are not:
		$$H_d = \sum_{\alpha\in T} h'_\alpha + \sum_{\alpha\notin T} h'_\alpha.$$
		
		Since $X_R$ acts nontrivially only on qubits in $R$, it commutes with all terms $h'_\alpha$ where $\alpha \notin T$. Consequently, the energy contribution from terms outside $T$ is invariant:
		$$ 
		\bra{0}^{\otimes n} X_R^\dagger \left(\sum_{\alpha\notin T} h'_\alpha\right) X_R \ket{0}^{\otimes n} 
		= 
		\bra{0}^{\otimes n} X_R^\dagger X_R \left(\sum_{\alpha\notin T} h'_\alpha\right) \ket{0}^{\otimes n}
		=
		\bra{0}^{\otimes n} \sum_{\alpha\notin T} h'_\alpha \ket{0}^{\otimes n}. 
		$$
		The energy difference is therefore determined solely by the local terms in $T$:
		\begin{align*}
			\Delta &= \Bigg\vert\bra{0}^{\otimes n} X_R^\dagger \left(\sum_{\alpha\in T} h'_\alpha\right) X_R \ket{0}^{\otimes n} - \bra{0}^{\otimes n} \sum_{\alpha\in T} h'_\alpha \ket{0}^{\otimes n}\Bigg\vert \\
			&= \Bigg\vert\bra{0}^{\otimes n} \left( \sum_{\alpha\in T} X_R^\dagger h'_\alpha X_R - h'_\alpha \right) \ket{0}^{\otimes n}\Bigg\vert.
		\end{align*}
		Using the fact that the expectation value of an operator is bounded by its spectral norm (also noting that $\Vert X_R^\dagger h'_\alpha X_R \Vert = \Vert h'_\alpha \Vert$), we bound $\Delta$:
		$$\Delta \leq \sum_{\alpha\in T} \left( \Vert X_R^\dagger h'_\alpha X_R \Vert + \Vert h'_\alpha \Vert \right) = 2 \sum_{\alpha\in T} \Vert h'_\alpha \Vert.$$
		
		To relate this sum to the quiet set properties, observe that summing the local energy $e(s)$ over all $s \in R$ counts every term in $T$ at least once (since every $\alpha \in T$ involves at least one qubit in $R$). Thus
		$$\sum_{\alpha\in T} \Vert h'_\alpha \Vert \le \sum_{s\in R} e(s).$$
		Since $R \subseteq Q_\delta$, every qubit $s \in R$ satisfies $e(s) \le \frac{\delta L}{n}$. Therefore
		$$\Delta \le 2 \sum_{s\in R} \frac{\delta L}{n} = \frac{2 |R| \delta L}{n}.$$
		By choosing $|R| \le r = \lfloor \frac{\mu \eta M n}{2\delta L} \rfloor$, we ensure
		$$\Delta \le \frac{2 \delta L}{n} \cdot \frac{\mu \eta M n}{2\delta L} = \mu \eta M.$$
		This implies $\vert\bra{\Phi_R} H_d \ket{\Phi_R} -  E_d\vert\le  \mu\eta M$, completing the proof. \(\qedhere\)
	\end{proof}
	
	\subsection{Bounding the Density of Low-Energy Eigenstates}\label{sec3}
	
	We now convert the existence of the family of orthonormal product states with bounded energy into a statement about the low-energy spectrum of $H_d$. While the states in $\mathcal{S}_r=\{\ket{\Phi_R}: |R|\le r\}$ are not necessarily eigenstates, their energy expectations are bounded. To utilize this, we invoke the Cauchy Interlacing Theorem.
	
	\begin{theorem}[Cauchy Interlacing Theorem~\cite{Horn_Johnson_1985}]\label{th: Cauchy}
		Let $A \in \mathbb{C}^{N\times N}$ be a Hermitian matrix with eigenvalues $\alpha_1\leq\alpha_2\leq \dots\leq\alpha_N$. For any $k \in \{1,\dots , N \}$, consider a principal submatrix $A_k$ of order $k$ (obtained by projecting $A$ onto a $k$-dimensional subspace spanned by a subset of the orthonormal basis). Let its eigenvalues be $\beta_1\leq\beta_2\leq\dots\leq\beta_k$. Then
		$$\alpha_i\leq\beta_i$$holds for all $i \in \{1,\dots , k\}$. 
		In particular, the trace satisfies
		$$\sum_{i=1}^{k}\alpha_i \leq \sum_{i=1}^{k}\beta_i = \operatorname{Tr}(A_k).$$
	\end{theorem}
	
	\begin{lemma}\label{lem: Cauchy}
		
		Let $\mathcal{S}_r=\{\ket{\Phi_R} : |R|\le r\}$ be the family of orthonormal states identified in Lemma~\ref{lem: Energy}, satisfying
		\begin{align*}
			\vert\bra{\Phi_R} H_d \ket{\Phi_R} - E_d\vert \le \mu\eta M,
		\end{align*}
		where $\mu>0$ and $\eta\in[0,1].$ Then,
		$$ 
		N(E_d + \mu M) \geq  \left(\frac{(1-\eta)\ \mu}{\mu+2}\right) \vert\mathcal{S}_r\vert .
		$$
	\end{lemma}
	
	\begin{proof}
		
		Extend $\mathcal{S}_r$ to a complete orthonormal basis with its elements as the first $|\mathcal{S}_r|$ vectors. Applying the Cauchy interlacing theorem to the $|\mathcal{S}_r|\times|\mathcal{S}_r|$ principal submatrix of $H_d$ on $\mathrm{span}(\mathcal{S}_r)$, we obtain diagonal entries $\bra{\Phi_R} H_d \ket{\Phi_R}$. Thus, we have the trace inequality:
		\begin{align*}\label{eq:trace_bound}
			\sum_{i=0}^{|\mathcal{S}_r|-1}\lambda_i \le \sum_{\ket{\Phi_R}\in\mathcal{S}_r}\bra{\Phi_R} H_d \ket{\Phi_R} \le |\mathcal{S}_r|(E_d+\mu\eta M).
		\end{align*}
		
		We now lower-bound the sum of the first $|\mathcal{S}_r|$ eigenvalues. Let $D \coloneqq N(E_d + \mu M)$ be the count of eigenvalues below the threshold. We have
		\begin{align*}
			\sum_{i=0}^{|\mathcal{S}_r|-1}\lambda_i = \sum_{i=0}^{D-1}\lambda_i + \sum_{i=D}^{|\mathcal{S}_r|-1}\lambda_i \ge D \lambda_0 + (|\mathcal{S}_r|-D)(E_d + \mu M).
		\end{align*}
		Combining these bounds we get
		\begin{align*}
			|\mathcal{S}_r|(E_d+\mu\eta M) &\geq  D \lambda_0 + (|\mathcal{S}_r|-D)(E_d + \mu M) \\
			D(E_d-\lambda_0+\mu M) &\geq  ( 1-\eta )\mu M |\mathcal{S}_r| \\
			D &\geq \frac{(1-\eta)\mu M}{E_d-\lambda_0+\mu M}\ |\mathcal{S}_r| \\
			D &\geq \frac{(1-\eta)\mu }{\frac{E_d-\lambda_0}{M}+\mu}\ |\mathcal{S}_r|\,.
		\end{align*}
		
		By the triangle inequality, the spectral norm satisfies $\|H_d\| \le \sum \|h'_\alpha\| = M$. All energy eigenvalues lie within $[-M, M]$, which gives us $E_d-\lambda_0\leq 2M.$ Thus, we have
		$$
			D \geq \frac{(1-\eta)\mu}{\mu+2}\ |\mathcal{S}_r|.\qedhere
		$$
	\end{proof}
	
	We now combine the preceding results to establish the main result of this section: a lower bound on the cumulative spectral count of the effective Hamiltonian $H_d$.
	\begin{theorem}\label{th: Bound}
		Let $H$ be an $n$-qubit $k$-local Hamiltonian with total interaction strength $M$, and let $E_d$ be defined as in Equation~(\ref{eq:ed_def}).
		Let $H_d = (U_d^*)^\dagger H U_d^*$ be the effective $K$-local Hamiltonian generated by an optimal depth-$d$ circuit. For any $\mu>0$, the cumulative spectral count $N(E_d +\mu M)$ of $H_d$ satisfies
		
		$$N(E_d + \mu M)
		=
		\Omega^*\!\left(
		\max_{\substack{\delta\geq 1+\frac{\mu M}{L} \\ \eta\in[0,1]}}
		\left\{
		\frac{\mu\ (1-\eta)}{\mu+2}\
		2^{\,\ent\!\left(\frac{\mu\eta M}{2(\delta-1)L}\right)\frac{\delta-1}{\delta}n}
		\right\}
		\right),$$
		where 
        $L$ be defined as in Equation~(\ref{eq:defL}).
	\end{theorem}
	
	\begin{proof}
		Fix parameters $\delta \ge 1 + \frac{\mu M}{L}$ and $\eta \in [0,1].$ By Lemma~\ref{lem: Energy}, setting the subset size limit to
		$r = \left\lfloor \frac{\mu \eta M n}{2\delta L} \right\rfloor$
		ensures that for every subset $R \subseteq Q_\delta$ with $\vert R\vert\le r$, the perturbed state satisfies $\vert\bra{\Phi_R} H_d \ket{\Phi_R} - E_d\vert \le \mu\eta M.$ By Lemma~\ref{lem: Orthgonality}, the family $\mathcal{S}_r=\{\ket{\Phi_R} : |R|\le r\}$ consists of mutually orthogonal states. Substituting the chosen value of $r$ into the cardinality bound from Corollary~\ref{cor: Sr Lower} yields
		$$|\mathcal{S}_r|
		=
		\Omega^*\!\left(
		2^{\,\ent\!\left(\frac{\mu\eta M}{2(\delta-1)L}\right)\frac{\delta-1}{\delta}n}
		\right).$$
		
		Finally, Lemma~\ref{lem: Cauchy} relates the size of this orthonormal family to the low-energy spectrum of $H_d$:
		$$
		N(E_d + \mu M) \geq  \frac{(1-\eta)\ \mu}{\mu+2} \vert\mathcal{S}_r\vert .
		$$
		Optimizing over all admissible values of $\delta$ and $\eta$ completes the proof. \(\qedhere\)
	\end{proof}
	
	As established earlier, since $U_d^*$ is a unitary operator, the Hamiltonians $H_d = (U_d^*)^\dagger H U_d^*$ and $H$ are isospectral. Hence, the lower bound for $N(E_d+\mu M)$ with respect to $H_d$ in Theorem~\ref{th: Bound} applies identically to the cumulative spectral count of the original Hamiltonian $H$. By taking $\eta=1/2$ and $\delta=2$, we obtain the following simpler lower bound:
	
	\begin{corollary}\label{cor:main}
        Let $H$ be an $n$-qubit $k$-local Hamiltonian with total interaction strength $M$, and let $E_d$ be defined as in Equation~(\ref{eq:ed_def}). For any $\mu =\Omega(1/\mathrm{poly}(n))$, the cumulative spectral count of $H$ satisfies:
		$$N(E_d + \mu M) = \Omega^* \left( 2^{\ent\left(\frac{\mu }{2^{d+2} k}\right)\frac{n}{2}} \right).$$
	\end{corollary}
	
	\begin{proof}
		We substitute specific parameters into the bound from Theorem~\ref{th: Bound}. Setting $\eta=1/2$ and $\delta=2$ yields an exponent of $\ent\left(\frac{\mu M}{4L}\right)\frac{n}{2}$. 
        
        Recall that $L = \sum_{s=1}^n \sum_{\alpha \in \mathcal{I}(s)} \|h'_\alpha\|$. Since each of the $m$ interaction terms $h'_\alpha$ in $H_d$ acts on at most $K$ qubits, summing over all sites counts each interaction term at most $K$ times. Thus, $L \le K \sum_{\alpha} \|h'_\alpha\| = KM$. Because the binary entropy function $\ent(x)$ is monotonically increasing for $x \in [0, 1/2]$, replacing $L$ with its upper bound $KM$ strictly maintains the inequality. Substituting $K = 2^d k$, the argument of the entropy function becomes
		$$ \frac{\mu M}{4(KM)} = \frac{\mu}{4K} = \frac{\mu}{2^{d+2}k}. $$
		Finally, because $\mu =\Omega(1/\mathrm{poly}(n))$, the prefactors $\frac{\mu(1-\eta)}{\mu+2}$ are entirely absorbed by the $\Omega^*(\cdot)$ notation,\footnote{Here we subtly make the assumption $\mu\leq 1$, which is true for large $n$.} yielding the final bound. \(\qedhere\)
	\end{proof}

\section{Faster Low-Energy Estimation and State Preparation}\label{bigsec3}

Leveraging the exponential density of states from Corollary~\ref{cor:main}, we now construct quantum algorithms for low-energy estimation and state preparation. 

We utilize the low-energy state preparation algorithm introduced in~\cite[Section 6]{lintong}. To prepare with high probability a state $\ket{\Psi}$ satisfying $\bra{\Psi} \tilde{H} \ket{\Psi} \le x$ for some Hamiltonian $\tilde{H}$ and target $x \ge 0$, the algorithm requires an initial state $\ket{\phi_0}$ that has an overlap of at least $\gamma$ with the low-energy subspace of $\tilde{H}$ corresponding to energies below $x - y$, for some $y>0$. The algorithm in \cite{lintong} assumes access to a unitary $U_I$ that prepares the initial state $\ket{\phi_0}$, and a block-encoding $U_{\tilde{H}}$ of the Hamiltonian. Treating these as black boxes, the query complexity is bounded by $O^*\left(\frac{1}{y\sqrt{\gamma}}\right)$ calls to $U_{\tilde{H}}$, and $O^*\left(\frac{1}{\sqrt{\gamma}}\right)$ calls to $U_I$.

In our setting, we consider the extended $2n$-qubit Hamiltonian $\tilde{H} = H \otimes I_{\text{anc}}$, where $H$ is our original $n$-qubit $k$-local Hamiltonian acting on the ``system'' register, and $I_{\text{anc}}$ is the identity matrix acting on an $n$-qubit ``ancilla'' register. Since $H$ is $k$-local, it is sparse and admits an efficient block-encoding. Extending $H$ to $\tilde{H}$ preserves the sparsity. This implies that each query to block-encoding $U_{\tilde{H}}$ can be implemented efficiently, in time $O(\mathrm{poly}(n))$. 

For the initial state we choose the \emph{maximally entangled state} on $2n$ qubits:
$$ \ket{\phi_0} = \frac{1}{\sqrt{2^n}} \sum_{k=0}^{2^n-1} |k\rangle_{\text{sys}} \otimes |k\rangle_{\text{anc}}, $$
where $\{\ket{k}\}_{k=0}^{2^n-1}$ is any orthonormal basis. This state can be efficiently prepared using a polynomial-depth circuit of Hadamard and CNOT gates, meaning the time cost of querying $U_I$ once is $O(\mathrm{poly}(n))$ and can be asymptotically neglected. Consequently, the overall time complexity of the algorithm matches its query complexity.

\begin{lemma}\label{lem: final} 
	Let $H$ be an $n$-qubit $k$-local Hamiltonian, $E_d$ be defined as in Equation~(\ref{eq:ed_def}), and $\tilde{H} = H \otimes I_{\text{anc}}$ be the extended Hamiltonian on $2n$ qubits. If $\ket{\phi_0}$ is the maximally entangled state on $2n$ qubits, its overlap $\gamma$ with the low-energy subspace of $\tilde{H}$ corresponding to energies below $E_d+\mu M$ satisfies
	$$ \gamma = \Omega^* \left( 2^{\ent\left(\frac{\mu }{2^{d+2}\ k}\right)\frac{n}{2}-n} \right)\,. $$
\end{lemma}
\begin{proof}
	Let $H = \sum_{i=0}^{2^n-1} \lambda_i \ket{\lambda_i}\bra{\lambda_i}$ be the spectral decomposition of $H$. The eigenstates of $\tilde{H} = H \otimes I_{\text{anc}}$ are $\ket{\lambda_i}_{\text{sys}} \otimes \ket{x}_{\text{anc}}$ with identical eigenvalues $\lambda_i$, for any orthonormal ancilla basis state $\ket{x}$. Let $\Pi_{\le E_d+\mu M}$ be the projector onto the subspace of $\tilde{H}$ with energies at most $E_d+\mu M$. We can write this projector as $P_{\le E_d+\mu M} \otimes I_{\text{anc}}$, where $P_{\le E_d+\mu M} = \sum_{i : \lambda_i \le E_d+\mu M} \ket{\lambda_i}\bra{\lambda_i}$.
	
	We expand the initial state $\ket{\phi_0}$ in the eigenbasis of $H$:
	$$ \ket{\phi_0} = \frac{1}{\sqrt{2^n}} \sum_{i=0}^{2^n-1} \ket{\lambda_i}_{\text{sys}} \otimes \ket{\lambda_i}_{\text{anc}}. $$
	The overlap between $\ket{\phi_0}$ and the subspace of energies below $E_d+\mu M$ satisfies:
	\begin{align*}
		\gamma & =  \bra{\phi_0}\Pi_{\le E_d+\mu M} \ket{\phi_0} \\
		& = \bra{\phi_0}(P_{\le E_d+\mu M} \otimes I_{\text{anc}}) \ket{\phi_0} \\
		& = \frac{1}{2^n} \sum_{i : \lambda_i \le E_d+\mu M} 1.
	\end{align*}
	This sum is exactly the number of eigenstates of $H$ with energy below $E_d+\mu M$. By Corollary~\ref{cor:main}, we can bound this density of states:
	$$ \gamma = \frac{\Omega^* \left(  2^{\ent\left(\frac{\mu }{2^{d+2}\ k}\right)\frac{n}{2}} \right)}{2^n} = \Omega^* \left( 2^{\ent\left(\frac{\mu }{2^{d+2}\ k}\right)\frac{n}{2}-n} \right). \qedhere $$
\end{proof}

	The exponential overlap with the low-energy subspace guaranteed by Lemma~\ref{lem: final} helps us achieve a faster quantum algorithm for low-energy state preparation. As a direct consequence of this state preparation algorithm, we can estimate this low-energy value with inverse-polynomial precision without altering the asymptotic runtime: 

    \begin{reptheorem}
        \maintheoremtext
    \end{reptheorem}

\begin{proof}
	\noindent\textbf{(i) Low-Energy State Preparation:}
	We apply the low-energy state preparation algorithm from~\cite[Section 6]{lintong} to the extended $2n$-qubit Hamiltonian $\tilde{H} = H \otimes I_{\text{anc}}$, where $H$ is our $n$-qubit $k$-local Hamiltonian acting on the $n$ system qubits and $I_{\text{anc}}$ is the identity on the $n$ ancilla qubits. Further, we take the maximally entangled state on $2n$ qubits as the initial state.
    
	
	Set
	$$x = E_d + \varepsilon M, \qquad y = \frac{\varepsilon M}{n}, \qquad \mu = \left(1 - \frac{1}{n}\right)\varepsilon.$$
	Then
	$$x - y = E_d + \varepsilon M - \frac{\varepsilon M}{n} = E_d + \mu M.$$
	
	By Lemma~\ref{lem: final}, the maximally entangled state $\ket{\phi_0}$ has an overlap of at least $\gamma$ with the subspace spanned by eigenstates of $\tilde{H}$ with energies at most $E_d + \mu M = x - y$. Hence, the input conditions of the algorithm are satisfied. Therefore, with high probability, the algorithm outputs a $2n$-qubit pure state $\ket{\Psi}$ such that
	$$\bra{\Psi} \tilde{H} \ket{\Psi} = \bra{\Psi} (H \otimes I_{\text{anc}}) \ket{\Psi} \le x = E_d + \varepsilon M.$$
	
	To obtain a state on the original $n$-qubit system, we trace out the ancilla register. Let $\rho = \operatorname{Tr}_{\text{anc}}[\ket{\Psi}\bra{\Psi}]$ be the reduced density matrix on the system qubits. By the properties of the partial trace\footnote{Recall that by the definition of the partial trace, we have $\operatorname{Tr}[(M_A\otimes I_B)\rho_{AB}]=\operatorname{Tr}_A[M_A\cdot\operatorname{Tr}_B[\rho_{AB}]]$.}, we get
	$$ \bra{\Psi} (H \otimes I_{\text{anc}}) \ket{\Psi} = \operatorname{Tr}[(H \otimes I_{\text{anc}})\ket{\Psi} \bra{\Psi}] = \operatorname{Tr}_{\text{sys}}[H\cdot\operatorname{Tr}_{\text{anc}}[\ket{\Psi} \bra{\Psi}]] = \operatorname{Tr}_{\text{sys}}[H\rho].$$
	This implies
	$$\operatorname{Tr}_{\text{sys}}[H \rho] \le E_d + \varepsilon M.$$
	Thus, we have successfully prepared an $n$-qubit mixed state $\rho$ satisfying the target energy bound.
	
	The running time of the algorithm is bounded by the total number of queries to $U_I$ and $U_{\tilde{H}}$:
	$$
	O^*\!\left(\frac{1}{\sqrt{\gamma}}\right) + O^*\!\left(\frac{1}{y\sqrt{\gamma}}\right) 
	= 
	O^*\!\left(\frac{1}{\sqrt{\gamma}}\right) + O^*\!\left(\frac{n}{\varepsilon M \sqrt{\gamma}}\right) 
	= 
	O^*\!\left(\frac{1}{\sqrt{\gamma}}\right),
	$$
	where the last step follows from the fact that $\varepsilon = \Omega(1/\mathrm{poly}(n))$. Substituting the bound on $\gamma$ from Lemma~\ref{lem: final}, we obtain the final time complexity
	$$T_{\text{prep}} = O^*\!\left( 2^{\frac{n}{2}\left(1 - \frac{1}{2}\ent\left(\frac{\varepsilon}{2^{d+2}k}\right)\right)} \right).$$
	Because the lowest possible energy in the system is the ground state energy $\lambda_0$, the expected energy of the prepared state $\rho$ is strictly guaranteed to lie within the interval $[\lambda_0,\ E_d+\varepsilon M]$.

    \paragraph*{(ii) Low-Energy Estimation:} To estimate the energy of the prepared state, we rely on the linearity of expectation values. We decompose the Hamiltonian into its local interaction terms $H = \sum_{\alpha=1}^m h_\alpha$. The total energy expectation is given by 
	$$ \operatorname{Tr}[H\rho] = \sum_{\alpha=1}^{m}\operatorname{Tr}[h_\alpha\rho]. $$
	
	We utilize the state preparation algorithm from part (i) of this proof as a subroutine. Since every successful run prepares the system register in this fixed state $\rho$, we can independently sample the local terms $h_\alpha$ across different successful runs.
	
	To achieve an inverse-polynomial precision $\varepsilon = \Omega(1/\mathrm{poly}(n))$ for the total energy estimate, we require a number of samples $S = O(\mathrm{poly}(n))$ for each local term. Because the Hamiltonian consists of $m = O(\mathrm{poly}(n))$ local terms, the total number of successful state preparations required is $m \times S = O(\mathrm{poly}(n))$. 
	
	Let $p_{\text{succ}}$ be the success probability of a single run of the state preparation algorithm in part (i). The expected time to achieve one successful preparation is $T_{\text{prep}} / p_{\text{succ}}$. Because the algorithm guarantees success with high probability, $1/p_{\text{succ}}$ is bounded by a constant. Therefore, the total expected runtime to gather all necessary samples for the estimation is
	$$ T_{\text{est}} = O(\mathrm{poly}(n)) \times T_{\text{prep}}. $$
	This multiplicative polynomial overhead is asymptotically absorbed by the $O^*(\cdot)$ notation. Thus, the low-energy estimation is achieved with the same time complexity as the state preparation. Consequently, 
    $$T_{\text{est}} = O^*\!\left( 2^{\frac{n}{2}\left(1 - \frac{1}{2}\ent\left(\frac{\varepsilon}{2^{d+2}k}\right)\right)} \right). \qedhere$$
\end{proof}

	\section{Conclusion}\label{bigsec4}

    While the square-root barrier of unstructured Grover search has recently been broken for general $k$-local Hamiltonians \cite{buhr1}, we have demonstrated that stricter asymptotic improvements are achievable for the class of Hamiltonians with a low-depth ground state. By shifting the energy threshold from the ground energy $\lambda_0$ to the energy threshold achievable by low-depth circuits $E_d$, we have established a faster state preparation and energy estimation method for Hamiltonians whose ground energy is well-approximated by low-depth states. Our work proves that for such Hamiltonians, targeting the low-depth state energy allows for a runtime scaling governed by the binary entropy function, providing a speedup over the rational-function scaling established in \cite{buhr1}.
	
	The mechanism driving this speedup is the locality-preserving nature of low-depth circuits. Because these unitaries $U_d$ result in a bounded lightcone, the effective Hamiltonian $H_d = (U_d^*)^\dagger H U_d^*$ remains $K$-local for some constant $K$. This structural preservation allows us to identify an exponentially dense manifold of low-energy product states, which gives rise to an exponentially dense low-energy subspace of $H_d$. Since $H_d$ and $H$ are isospectral due to being unitarily equivalent, the low-energy subspace of $H$ is also exponentially dense. This property serves as the ``engine'' of our faster, low-energy state-preparation algorithm. A low-depth state can be fully described by its low-depth circuit recipe. Therefore, conceptually, our results highlight the distinction between strongly entangled quantum states and those with efficient classical descriptions.
	
	Looking ahead, several compelling directions for future research emerge from this work:
	\begin{itemize}
		\item Locality-Preserving Unitaries: A natural extension is to identify the full class of unitaries $U$, beyond low-depth circuits, that preserve the locality of $U^\dagger H U$. Any transformation that maintains a bounded support for the Hamiltonian local terms could potentially yield a similar asymptotic speedup.
		
		\item While our present work focuses on local perturbations of the all-zero state, a natural extension is to consider their effect on tensor network states, such as Matrix Product States (MPS). A key question is whether the exponential decay of correlations exhibited by these states can be leveraged to establish an analogous density-of-states argument.
	
	\end{itemize}

	In summary, by tailoring our algorithmic goals to the inherent structure of low-depth states, we have shown that the landscape of low-energy state preparation and estimation is more accessible than previously realized for a wide variety of quantum many-body systems.
        
	\section*{Acknowledgments}
    The authors are very grateful to Harry Buhrman, Sevag Gharibian, Zeph Landau and Norbert Schuch for helpful discussions.  
	FLG and RM are supported by JSPS KAKENHI grant Nos.~24H00071, 25K24674 and 25K24465, MEXT Q-LEAP grant No.~JPMXS0120319794, JST ASPIRE grant No.~JPMJAP2302 and JST CREST grant No.~JPMJCR24I4. ST is supported by JSPS KAKENHI grant No.~JP22K11909.
    \vspace{2em}
	
	\bibliographystyle{alpha}
	\bibliography{ref.bib}

\end{document}